\documentclass[]{article}

\usepackage{amsmath,amssymb}
\usepackage{bm,enumerate}
\usepackage{amsthm}
\usepackage{graphicx}
\usepackage{array} %
\usepackage{enumitem} %
\usepackage{verbatim} %
\usepackage{subfig} %
\usepackage{algorithm}
\usepackage{algpseudocode}
\usepackage{fancyhdr}
\usepackage{natbib}

\usepackage{booktabs} %
\usepackage{fancyvrb}
\usepackage{authblk}

\usepackage{footnote}
\usepackage{hyperref}

\newcommand{\eat}[1]{}

\floatstyle{ruled}
\newfloat{algorithm}{tb}{loa}
\floatname{algorithm}{Algorithm}
\newtheorem{theorem}{Theorem}
\newtheorem{corollary}[theorem]{Corollary}

\newtheorem{lemma}[theorem]{Lemma}

\newcommand{\defn}{:=}
\DeclareFontFamily{U}{mathx}{\hyphenchar\font45}
\DeclareFontShape{U}{mathx}{m}{n}{
      <5> <6> <7> <8> <9> <10>
      <10.95> <12> <14.4> <17.28> <20.74> <24.88>
      mathx10
      }{}
\DeclareSymbolFont{mathx}{U}{mathx}{m}{n}
\DeclareMathSymbol{\bigtimes}{1}{mathx}{"91}

\def\var{{\rm Var}}
\def\E{{\ensuremath{\mathbb E}}}

\long\def\comment#1{}

\newcommand{\ical}{\ensuremath{\mathcal I}}
\newcommand{\jcal}{\ensuremath{\mathcal J}}

\begin{document}

\title{Data Sketches for Disaggregated Subset Sum and Frequent Item Estimation}

\author{Daniel Ting}
\affil{Tableau Software}

\maketitle

\begin{abstract}
	We introduce and study a new data sketch for processing massive datasets. It addresses two common problems: 1) computing a sum given arbitrary filter conditions and 2) identifying the frequent items or heavy hitters in a data set. For the former, the sketch provides unbiased estimates with state of the art accuracy. It handles the challenging scenario when the data is disaggregated. In this case, a per unit metric of interest can only be computed as an expensive pre-aggregation of the raw, disaggregated data. For example, the metric of interest may be total clicks per user while the raw data is a click stream containing multiple rows per user. Thus the sketch is suitable for use in a wide range of applications including computing historical click through rates for ad prediction, reporting user metrics from event streams, and measuring network traffic for IP flows.

We prove and empirically show the sketch has good properties for both the disaggregated subset sum estimation and frequent item problems. On i.i.d. data, it not only picks out the frequent items but  gives strongly consistent estimates for the proportion of each frequent item. 
For subset sum estimation, it asymptotically draws a probability proportional to size sample that is optimal for estimating the sum over the data. For non i.i.d. data, we show that it typically does much better than random sampling for the frequent item problem and never does worse. For subset sum estimation, we show that even for pathological sequences, the variance is close to that of an optimal sampling design. Empirically, despite the disadvantage of operating on disaggregated data, our method matches or bests priority sampling, a state of the art method on pre-aggregated data. When compared to uniform sampling, it performs orders of magnitude better on skewed data. We also propose extensions to the sketch that allow it to be used in combining multiple data sets, in distributed systems, and for time decayed aggregation.

\end{abstract}

\section{Introduction}
When analyzing massive data sets, even simple operations such as computing a sum or mean are costly and time consuming. These simple operations are frequently performed both by people investigating the data interactively and asking a series of questions about it as well as in automated systems which must monitor or collect a multitude of statistics.

Data sketching algorithms enable the information in these massive datasets to be efficiently processed, stored, and queried. 
This allows them to be applied, for example, in real-time systems, both for ingesting massive data streams or for interactive analysis.

In order to achieve this efficiency, sketches are designed to only answer a specific class of question, and there is typically error in the answer.
In other words, it is a form of lossy compression on the original data where one must choose what to lose in the original data. A good sketch makes the most efficient use of the data so that the  errors are minimized while having the flexibility to answer a broad range of questions of interest.
Some sketches, such as HyperLogLog, are constrained to answer very specific questions with extremely little memory. On the other end of the spectrum, sampling based methods such as %
coordinated sampling \cite{brewer1972selecting}, \cite{cohen2013coordinated} are able to answer almost any question on the original data but at the cost of far more space to achieve the same approximation error. 

We introduce a sketch, Unbiased Space Saving, that simultaneously addresses two common data analysis problems: the disaggregated subset sum problem and the frequent item problem.  
This makes the sketch more flexible than previous sketches that address one problem or the other. Furthermore, it is efficient as it provides state of the art performance on the disaggregated subset sum problem. On i.i.d. streams it has a stronger provable consistency guarantee for frequent item count estimation than previous results, and on non-i.i.d. streams it performs well both theoretically and empirically. 
In addition, we derive an error estimator with good coverage properties that allows a user to assess the quality of a disaggregated subset sum result.

The disaggregated subset sum estimation is a more challenging variant of the subset sum estimation problem  \cite{duffield2007priority}, the extremely common problem of computing a sum or mean over a dataset with arbitrary filtering conditions. 
In the disaggregated subset sum problem \cite{cohen2007sketching}, \cite{gibbons1998new} the data is "disaggregated" so that a per item metric of interest is split across multiple rows. %
For example in an ad click stream, the data may arrive as a stream of single clicks that are identified with each ad while the metric of interest is the total number of clicks per ad.
The frequent item problem is the problem of identifying the heavy hitters or most frequent items in a dataset. 
Several sketches exist for both these individual problems.
In particular, the Sample and Hold methods of \cite{cohen2007sketching}, \cite{estan2003new}, \cite{gibbons1998new} address the disaggregated subset sum estimation problem. 
Frequent item sketches include the Space Saving sketch \cite{metwally2005spacesaving},
Misra-Gries sketch \cite{misra1982frequent}, and Lossy Counting sketch
\cite{manku2002approximate}. 

Our sketch is an extension of the Space Saving frequent item sketch, and as such,
has stronger frequent item estimation properties than Sample and Hold. 
In particular, unlike Sample and Hold, theorem \ref{thm:frequent item} gives both that a frequent item will eventually be included in the sketch with probability 1, and that the proportion of times it appears will be consistently estimated
for i.i.d. streams. In contrast to frequent item sketches which are biased, our Unbiased Space Saving sketch gives unbiased estimates for any subset sum, including subsets containing no frequent items. 

Our contributions are in three parts: 1) the development of the Unbiased Space Saving sketch, 2) the generalizations obtained from understanding the properties of the sketch and the mechanisms by which it works, and 3) the theoretical and empirical results establishing the correctness and efficiency of the sketch for answering the problems of interest. In particular, the generalizations allow multiple sketches to be merged so that information from multiple data sets may be combined as well as allowing it to be applied in distributed system. Other generalizations include the ability to handle signed and real-valued updates as well as time-decayed aggregation.
We empirically test the sketch on both synthetic and real ad prediction data.
Surprisingly, we find that it even outperforms priority sampling, a method that requires pre-aggregated data.

This paper is structured as follows. First, we describe the disaggregated subset sum problem, some of its applications, and related sketching problems. We then introduce our sketch, Unbiased Space Saving, as a small but significant modification of the Space Saving sketch. We examine its relation to other frequent item sketches, and show that they differ only in a "reduction" operation. This is used to show that any unbiased reduction operation yields an unbiased sketch for the disaggregated subset sum estimation problem. 
The theoretical properties of the sketch are then examined. We prove its consistency for the frequent item problem and for drawing a probability proportional to size sample. We derive a variance estimator and show that it can be used to generate good confidence intervals for estimates. Finally, we present experiments using real and synthetic data.

\section{Two Sketching Problems}	
\section{Disaggregated subset sum problem}
Many data analysis problems consist of a simple aggregation over some filtering and group by conditions.

\begin{center}
	\begin{BVerbatim}
	SELECT sum(metric), dimensions
	FROM table
	WHERE filters
	GROUP BY dimensions
	\end{BVerbatim}
\end{center}

This problem has several variations that depend on what is known about the possible queries and about the data before the sketch is constructed.
For problems in which there is no group by clause and the set of possible filter conditions are known before the sketch is constructed,
counting sketches such as the CountMin sketch \cite{cormode2005countmin} and 
AMS sketch \cite{alon1999space} are appropriate. 
When the filters and group by dimensions are not known and arbitrary, the problem is the subset sum estimation problem. Sampling methods such as priority sampling \cite{duffield2007priority} can be used to solve it. These work by exploiting a measure of importance for each row and sampling important rows with high probability. For example, when computing a sum, the rows containing large values contribute more to the sum and should be retained in the sample.

The disaggregated subset sum estimation problem is a more difficult variant where there is little to no information about row importance and only a small amount of information about the queries. For example, many user metrics, such as number of clicks, are computed as aggregations over some event stream where each event has the same weight 1 and hence, the same importance. 
Filters and group by conditions can be arbitrary except for a  small restriction that one cannot query at a granularity finer than a specified unit of analysis. In the click example, the finest granularity may be at the user level.
One is allowed to query over arbitrary subsets of users but cannot query a subset of a single user's clicks. The data is "disaggregated" since the relevant per unit metric is split across multiple rows. We will refer to something at the smallest unit of analysis as an {\em item} to distinguish it from one row in the data.

Since pre-aggregating to compute per unit metrics does not reduce the amount of relevant information, it follows that the best accuracy one can achieve is to first pre-aggregate and then apply a sketch for subset sum estimation. This operation, however, is extremely expensive, especially as the number of units is often large. Examples of units include users and ad id pairs for ad click prediction, source and destination IP pairs for IP flow metrics, and distinct search queries or terms. Each of these have trillions or more possible units.

Several sketches based on sampling have been proposed that address the disaggregated subset sum problem. These include the bottom-k sketch \cite{cohen2007bottomk} which samples items uniformly at random,
the class of "NetFlow" sketches \cite{estan2004building},
and the Sample and Hold sketches \cite{cohen2007sketching}, \cite{estan2003new}, \cite{gibbons1998new}. Of these, the Sample-and-Hold sketches are clearly the best as they use strictly more information than the other methods to construct samples and maintain aggregate statistics.
We describe them in more depth in section \ref{sec:sample and hold}.

The Unbiased Space Saving sketch we propose throws away even less information than previous sketches. Surprisingly, this allows it to match the accuracy of priority sampling, a nearly optimal subset sum estimation algorithm \cite{szegedy2006dlt},
which uses pre-aggregated data. In some cases, our sketch achieves better accuracy despite being computed on disaggregated data.

\subsection{Applications}
The disaggregated subset sum problem has many applications. These include machine learning and ad prediction \cite{shrivastava2016time},
analyzing network data \cite{estan2004building}, \cite{cohen2007sketching},
detecting distributed denial of service attacks \cite{sekar2006lads},
database query optimization and join size estimation \cite{vengerov2015join}, 
as well as analyzing web users' activity logs or other business intelligence applications.

For example, in ad prediction the historical click-through rate and other historical data are among the most powerful features for future ad clicks \cite{he2014practical}.
Since there is no historical data for newly created ads, one may use historical click or impression data for previous ads with similar attributes such as the same advertiser or product category \cite{richardson2007predicting}. 
In join size estimation, it allows the sketch to estimate the size under the arbitrary filtering conditions that a user might impose.

It also can be naturally applied to hierarchical aggregation problems.
For network traffic data, IP addresses are arranged hierarchically. A network administrator may both be interested in individual nodes that receive or generate an excess of traffic or aggregated traffic statistics on a subnet. Several sketches have been developed to exploit hierarchical aggregations including \cite{cormode2008finding}, \cite{mitzenmacher2012hierarchical}, and \cite{zhang2004online}. Since a disaggregated subset sum sketch can handle arbitrary group by conditions, it can  compute the next level in a hierarchy.

\subsection{Frequent item problem}
The frequent item or heavy hitter problem is  related to the disaggregated subset sum problem. Our sketch is an extension of Space Saving, \cite{metwally2005spacesaving}, a frequent item sketch. 
Like the disaggregated subset sum problem, frequent item sketches are computed with respect to a unit of analysis that requires a partial aggregation of the data.
However, only functions of the most frequent items are of interest. Most frequent item sketches are  deterministic and have  deterministic guarantees on both the identification of frequent items and the error in the counts of individual items.
However, since counts in frequent item sketches are biased, further aggregation on the sketch can lead to large errors when bias accumulates as shown in section \ref{sec:pathological}.

Our work is based on a frequent item sketch, but applies randomization to achieve unbiased count estimates. This allows it to be used in subset sum queries. Furthermore, it  maintains good frequent item estimation properties as proved in theorems \ref{thm:frequent item} and \ref{thm:worst case}.

\section{Unbiased Space-saving}
Our sketch is based on the Space Saving sketch \cite{metwally2005spacesaving} used in frequent item estimation.
We will refer to it as Deterministic Space Saving to differentiate it from our randomized sketch.
For simplicity, we consider the case where the metric of interest is the count for each item. The Deterministic Space Saving sketch works by maintaining a list of $m$ bins labeled by distinct items. A new row with item $i$ increments $i$'s counter if it is in the sketch. Otherwise, the smallest bin is incremented, and its label is changed to $i$.
Our sketch introduces one small modification. If $\hat{N}_{min}$ is the count for the smallest bin, then only change the label with probability $1/(\hat{N}_{min} + 1)$. This change provably yields unbiased counts as shown in theorem \ref{thm:unbiased}. Algorithm \ref{fig:algo} describes these Space Saving sketches more formally.

\begin{algorithm}
	\begin{itemize}
		\item Maintain an $m$ list of $(item, count)$ pairs initialized to have count 0.
		\item For each new row in the stream, let $x_{new}$ be its item and increment the corresponding counter if the item is in the list. Otherwise, find the pair $(x_{min}, \hat{N}_{min})$ with the smallest count. Increment the counter and replace the item label with $x_{new}$ with probability $p$.
		\item For the original Space Saving algorithm $p=1$. For unbiased count estimates $p= 1/(\hat{N}_{min}+1)$.
	\end{itemize}
	\caption{Space-Saving algorithms}
	\label{fig:algo}
\end{algorithm}

\begin{table}
	\begin{tabular}{c|l}
		Notation & Definition \\ \hline
		$t$ & Number of rows encountered or time\\
		$\hat{N}_i(t)$ & Estimate for item $i$ at time $t$ \\
		$\hat{N}_{min}(t)$ & Count in the smallest bin at time $t$ \\
		$n_i, n_{tot}$ & True count for item $i$ and total over all items \\
		$\hat{N}_S, n_S$ & Estimated and true total count of items in $S$ \\
		$\mathbf{N}, \mathbf{n}$ & Vector of estimated and true counts \\
		$p_i$ & Relative frequency $n_i / n_{tot}$ of item $i$ \\
		$m$ & Number of bins in sketch \\
		$Z_i$ & Binary indicator if item $i$ is a label in the sketch \\
		$\pi_i$ & Probability of inclusion $P(Z_i = 1)$ \\
		$C_S$ & Number of items from set $S$ in the sketch
	\end{tabular}
	\label{tbl:symbols}
	\caption{Table of symbols}
\end{table}

\begin{theorem}
	\label{thm:unbiased}
	For any item $x$, the randomized Space-Saving algorithm in figure \ref{fig:algo} gives an unbiased estimate of the count of $x$. 
\end{theorem}
\begin{proof}
	Let $\hat{N}_x(t)$ denote the estimate for the count of $x$ at time $t$
	and $\hat{N}_{min}(t)$ be the count in the smallest bin. 
	We show that the expected increment to $N_x(t)$ is 1 if $x$ is the next item and 0 otherwise.
	Suppose $x$ is the next item. If it is in the list of counters, then it is incremented by exactly 1. Otherwise, it incremented by $\hat{N}_{min}(t)+1$ with probability $1/(\hat{N}_{min}(t)+1)$ for an expected increment of $1$.
	Now suppose $x$ is not the next item. The estimated count $\hat{N}_x(t)$ can only be modified if $x$ is the label for the smallest count. It is incremented with probability $\hat{N}_x(t)/ (\hat{N}_x(t)+1)$. Otherwise $\hat{N}_x(t+1)$ is updated to $0$.
	This gives the update an expected increment of 
	$\E \hat{N}_x(t+1) - \hat{N}_x(t) = (\hat{N}_x(t)+1) \hat{N}_x(t) / (\hat{N}_x(t)+1)  - \hat{N}_x(t) = 0$
	when the new item is not $x$.
\end{proof}

We note that although given any fixed item $x$, the estimate of its count is unbiased, each stored pair often contains an overestimate of the item's count.
This occurs since any item with a positive count will receive a downward biased estimate of 0 conditional on it not being in the sketch. Thus, conditional on an item appearing in the list, the count must be biased upwards.

\section{Related sketches and further generalizations}
Although our primary goal is to demonstrate the usefulness of the Unbiased Space-Saving sketch, we also try to understand the mechanisms by which it works and use this understanding to find extensions and generalizations. Readers only interested in the properties of Unbiased Space Saving may skip to the next section. 

In particular, we examine the relationships between Unbiased Space Saving and existing deterministic frequent items sketches as well as its relationship with probability proportional to size sampling. We show that existing frequent item sketches all share the same structure as an exact increment of the count followed by a size reduction. This size reduction is implemented as an adaptive sequential thresholding operation which biases the counts. 
Our modification replaces the thresholding operation with a subsampling operation. This observation allows us to extend the sketch. This includes endowing it with an unbiased merge operation that can be used to combine datasets or in distributed computing environments.

The sampling design in the reduction step may also be chosen to give the sketch different properties. For example, time-decayed sampling methods may be used to weight  recently occurring items more heavily. If multiple metrics are being tracked, multi-objective sampling \cite{cohen2015multi} %
may be used.

\subsection{Probability proportional to size sampling}
\label{sec:PPS}
Our key observation in generalizing Unbiased Space Saving is that the choice of label is a sampling operation. In particular, this sampling operation chooses the item with probability proportional to its size. We briefly review probability proportional to size sampling and priority sampling as well as the Horvitz-Thompson estimator which allows one to unbias the sum estimate from any biased sampling scheme. Probability proportional to size sampling (PPS) is of special importance for sampling for subset sum estimation as it is essentially optimal. Any good sampling procedure mimics PPS sampling.

For unequal probability samples, an unbiased estimator for the sum over the true population $\{x_i\}$ is given by the Horvitz-Thomson estimator 
$\hat{S} = \sum_i  \frac{x_i Z_i}{\pi_i} $
where $Z_i$ denotes whether $x_i$ is in the sample and $\pi_i = P(Z_i = 1)$ is the inclusion probability. When only linear statistics of the sampled items are computed, the item values may be updated $x^{new}_i = x_i / \pi_i$. 

When drawing a sample of fixed size, it is trivial to see that an optimal set of inclusion probabilities is given by $\pi_i \propto x_i$ when this is possible. In other words, it generates a probability proportional to size (PPS) sample.
In this case, each term in the sum is constant, so that the estimator is exact and has zero variance. 
When the data is skewed, drawing a truly probability proportional size sample may be impossible for sample sizes greater than 1. For example, given values $1,1,$ and $10$, any scheme to draw 2 items with probabilities exactly proportional to size has inclusion probabilities bounded by $1/10, 1/10,$ and $1$. The expected sample size is at most $12/10 < 2$. In this case, one often chooses inclusion probabilities $\pi_i = \min\{\alpha x_i, 1\}$ for some constant $\alpha$. The inclusion probabilities are proportional to the size if the size is not too large and 1 otherwise. 

Many algorithms exist for generating PPS samples. We briefly describe two as they are necessary for the merge operation given in section \ref{sec:merge}.
The splitting procedure of \cite{Deville1998SamplingSplitting} provides a class of methods to generate a fixed size PPS sample with the desired inclusion probabilities. Another method which approximately generates a PPS sample is priority sampling. Instead of exact inclusion probabilities which are typically intractable to compute, priority sampling generates a set of pseudo-inclusion probabilities. 

The splitting procedure is based on a simple recursion. At each step, 
the target distribution is split into a mixture of two simpler distributions. One flips a coin and based on the result, chooses to sample from one of the two simpler distribution. More formally, given a target vector of inclusion probabilities $\pi$ and two vectors of probabilities $\pi^{(0)}$ and  $\pi^{(1)}$ with $\pi = \alpha \pi^{(0)} + (1-\alpha) \pi^{(1)}$,
then drawing $D \sim Bernoulli(1-\alpha)$ and then drawing a sample with marginal inclusion probabilities $\pi^{(D)}$ gives a sample with inclusion probabilities matching the target $\pi$. There is great flexibility in choosing how to split, and when the split yields inclusion probabilities equal to 0 or 1, the subsequent sampling becomes easier. 

Priority sampling is a method that approximately draws a PPS sample. It generates a random priority $R_{i} = U_i / n_i$ for 
an item $i$ with value $n_i$. The values corresponding to the $m$ smallest priorities form the sample. Surprisingly, by defining the threshold $\tau$ be the $(m+1)^{th}$ smallest priority, it can be shown that for almost any function of just the samples, the expected value under this sampling scheme is the same as the expected value under independent $Bernoulli( \min\{1, n_i \tau_i\} )$ sampling. 

\subsection{Misra-Gries and  frequent item sketches}

The Misra-Gries sketch \cite{misra1982frequent}, \cite{demaine2002frequency}, \cite{karp2003simple} is a frequent item sketch and is isomorphic to the Deterministic Space Saving sketch \cite{agarwal2013mergeable}.
The only difference is that it decrements all counters rather than incrementing the smallest bin when processing an item that is not in the sketch. Thus, the count in the smallest bin for the Deterministic Space Saving sketch is equal to the total number of decrements in the Misra-Gries sketch.
Given estimates $\mathbf{\hat{N}}$ from a Deterministic Space Saving sketch,  
the corresponding estimated item counts for the Misra-Gries sketch are $\hat{N}^{MG}_i = (\hat{N}_i - \hat{N}_{min})_+$ where $\hat{N}_{min}$ is the count for the smallest bin and the operation $(x)_+$ truncates negative values to be 0. In other words, the Misra-Gries estimate is the same as the Deterministic Space Saving estimate soft thresholded by $\hat{N}_{min}$.
Equivalently, the Deterministic Space Saving estimates are obtained by adding back the total number of decrements $\hat{N}_{min}$ to any nonzero counter in the Misra-Gries sketch. 

The sketch has a deterministic error guarantee.
When the total number of items is $n_{tot}$ then the error for any item is at most $n_{tot} / m$.

Other frequent item sketches include the deterministic lossy counting and randomized sticky sampling sketches \cite{manku2002approximate}. We describe only lossy counting as sticky sampling has both worse practical performance and weaker guarantees than other sketches. 

A simplified version of Lossy counting applies the same decrement reduction as the Misra-Gries sketch but decrements occur at a fixed schedule rather than one which depends on the data itself. To count items with frequency $> N/m$, all counters are decremented after every $m$ rows. Lossy counting does not provide a guarantee that the number of counters can be bounded by $m$. In the worst case, the size can grow to $m \log (N/m)$ counters. Similar to the isomorphism between the Misra-Gries and Space-saving sketches, the original Lossy counting algorithm is recovered by adding the number of decrements back to any nonzero counter.

\subsection{Reduction operations}
Existing deterministic frequent item sketches differ in only the operation to reduce the number of nonzero counters. They all have the form described in algorithm \ref{alg:reduction} and have reduction operations that can be expressed as a thresholding operation. Although it is isomorphic to the Misra-Gries sketch, Deterministic Space Saving's reduction operation can also be described as collapsing the two smallest bins by adding the larger bin's count to the smaller one's. 

\begin{algorithm}
	\begin{itemize}
		\item Maintain current estimates of counts $\hat{\mathbf{N}}(t)$
		\item Increment $\hat{N}'_{x_{t+1}}(t+1) \gets \hat{N}_{x_{t+1}}(t) + 1$.
		\item $\hat{\mathbf{N}}(t+1) \gets ReduceBins(\hat{\mathbf{N}'}(t+1), t+1)$		
	\end{itemize}
	\caption{General frequent item sketching}
	\label{alg:reduction}
\end{algorithm}

Modifying the reduction operation provides the sketch with different properties. We highlight several uses for alternative reduction operations. 

The reduction operation for Unbiased Space Saving can be seen as a PPS sample on the two smallest bins. 
A natural generalization is to consider a PPS sample on all the bins.
We highlight three benefits of such a scheme.
First, items can be added with arbitrary counts or weights. Second, the sketch size can be reduced by multiple bins in one step. Third, there is less quadratic variation added by one sampling step, so error can be reduced. The first two benefits are obvious consequences of the generalization. To see the third,
consider when a new row contains an item not in the sketch, and let $\jcal$ be the set of bins equal to $\hat{N}_{min}$.
When using the thresholded PPS inclusion probabilities from section \ref{sec:PPS}, the resulting PPS sample has inclusion probability $\alpha = |\jcal|/(1 + |\jcal| \hat{N}_{min})$ for the new row's item and $\alpha \hat{N}_{min}$ for bins in $\jcal$. Other bins have inclusion probability $1$. After sampling, the Horvitz-Thompson adjusted counts are $1/ |\jcal| + \hat{N}_{min}$. Unbiased Space Saving is thus a further randomization to convert the real valued increment $1/|\jcal|$ over $|\jcal|$ bins to an integer update on a single bin. Since Unbiased Space Saving adds an additional randomization step, the PPS sample has smaller variance. The downside of this procedure, however, is that it requires real valued counters that require more space per bin.
The update cost when using the stream summary data structure \cite{metwally2005spacesaving} remains $O(1)$.

Changing the sampling procedure can also provide other desirable behaviors. 
Applying forward decay sampling \cite{cormode2009forward}
allows one to obtain estimates that weight recent items more heavily. Other possible operations include adaptively varying the sketch size in order to only remove items with small estimated frequency. 

Furthermore, the reduction step does not need to be limited strictly to subsampling. Theorem \ref{thm:unbiased reduction} gives that any unbiased reduction operation yields unbiased estimates. This generalization allows us to analyze Sample-and-Hold sketches.

\begin{theorem}
	\label{thm:unbiased reduction}
	Any reduction operation where the expected post-reduction estimates are equal to the pre-reduction estimates yields an unbiased sketch for the disaggregated subset estimation problem. More formally, if $\E(\hat{\mathbf{N}}(t) | S_{pre}(t)) = \hat{\mathbf{N}}_{pre}(t)$ where $S_{pre}(t), \hat{\mathbf{N}}_{pre}(t)$ are the sketch and estimated counts before reduction at time step $t$ and $\hat{\mathbf{N}}(t)$ is the post reduction estimate, then $\hat{\mathbf{N}}(t)$ is an unbiased estimator.
\end{theorem}
\begin{proof}
	Since $\hat{\mathbf{N}}_{pre}(t) = \hat{\mathbf{N}}_{post}(t-1) + (\mathbf{n}(t) - \mathbf{n}(t-1))$,	
	it follows that $\hat{\mathbf{N}}(t) - \mathbf{n}(t)$ is a martingale with respect to the filtration adapted to $S(t)$. Thus, $\E \hat{\mathbf{N}}(t) = \mathbf{n}(t)$, and the sketch gives unbiased estimates for the disaggregated subset sum problem. 
\end{proof}

We also note that reduction operations can be biased. 
The merge operation on the Misra-Gries sketch given by \cite{agarwal2013mergeable} performs a soft-thresholding by the size of the $(m+1)^{th}$ counter rather than by 1. This also allows it to reduce the size of the sketch by more than 1 bin at a time. It can be modified to handle deletions and arbitrary numeric aggregations by making the thresholding operation two-sided so that negative values are shrunk toward 0 as well. In this case, we do not provide a theoretical analysis of the properties. 

Modifying the reduction operation also yields interesting applications outside of counting. In particular, a reduction operation on matrices can yield accurate low rank decompositions \cite{liberty2013matrixsketching}, \cite{ghashami2006sparsematrixsketching}.

\subsection{Sample and Hold}
\label{sec:sample and hold}
To our best knowledge, the current state of the art sketches designed to answer disaggregated subset sum estimation problems are the family of sample and hold sketches \cite{gibbons1998new}, \cite{estan2003new}, \cite{cohen2007sketching}. These methods can also be described with a randomized reduction operation.

For adaptive sample and hold \cite{cohen2007sketching}, the sketch maintains an auxiliary variable $p$ which represents the sampling rate. Each point in the stream is assigned a $U_i \sim Uniform(0,1)$ random variable, and the items in the sketch are those with $U_i < p$. If an item remains in the sketch starting from time $t_0$, then the counter stores the number of times it appears in the stream after the initial time. 
Every time the sketch becomes too large, the sampling rate is decreased so that under the new rate $p'$, one item is no longer in the sketch. 

It can be shown that unbiased estimates can be obtained by keeping a counter value the same with probability $p'/p$ and decrementing the counter by a random $Geometric(p')$ random variable otherwise. If a counter becomes negative, then it is set to 0 and dropped. Adding back the mean $(1-p')/p'$ of the $Geometric$ random variable to the nonzero counters gives an unbiased estimator. Effectively, the sketch replaces the first time an item enters the sketch with the expected $Geometric(p')$ number of tries before it successfully enters the sketch and it adds the actual count after the item enters the sketch. Using the memoryless property of $Geometric$ random variables, it is easy to show that the sketch satisfies the conditions of theorem \ref{thm:unbiased reduction}. It is also clear that one update step adds more error than Unbiased Space Saving 
as it potentially adds $Geometric(p')$ noise with variance $(1-p')/p'^2$ to every bin. Furthermore, the eliminated bin may not even be the smallest bin. Since $p'$ is the sampling rate, it is expected to be close to 0. 
By contrast, Unbiased Space Saving has bounded increments of $1$
for bins other than the smallest bin, and the only bin that can be removed is the current smallest bin.

The discrepancy is especially prominent for frequent items.
A frequent item in an i.i.d. stream for Unbiased Space Saving enters the sketch almost immediately, and the count for the item is nearly exact as shown in theorem \ref{thm:frequent item}. For adaptive sample and hold, the first $n_i(1-p')$ occurrences of item $i$ are expected to be discarded and replaced with a high variance $Geometric(p')$ random variable. Since $p'$ is typically small in order to keep the number of counters low, most of the information about the count is discarded.

Another sketch, step sample-and-hold, avoids the problem by maintaining counts for each "step" when the sampling rate changes. However, this is more costly both from storage perspective as well as a computational one. For each item in the sketch, computing the expected count takes time quadratic in the number of steps $J_i$ in which the step's counter for the item is nonzero, and storage is linear in $J_i$.

\subsection{Merging and Distributed counting}
\label{sec:merge}
The more generalized reduction operations allow for merge operations on the sketches.
Merge operations and mergeable sketches \cite{agarwal2013mergeable} are important since they allow a collection of sketches, each answering questions about the subset of data it was constructed on, to be combined to answer a question over all the data. 
For example, a set of frequent item sketches that give trending news for each country  can be combined to give trending news for Europe  as well as a multitude of other possible combinations. Another common scenario arises when sketches are aggregated across time. Sketches for clicks may be computed per day, but the final machine learning feature may combine the last 7 days. 

Furthermore, merges make sketches  more practical to use in real world systems.
In particular, they allow for simple distributed computation. In a map-reduce framework, each mapper can quickly compute a sketch, and only a set of small sketches needs to be sent over the network to perform an aggregation at the reducer. 

As noted in the previous section, the Misra-Gries sketch has a simple merge operation which preserves its deterministic error guarantee. It simply soft thresholds by the $(m+1)^{th}$ largest counter so that at most $m$ nonzero counters are left. Mathematically, this is expressed as  $\hat{N}^{new}_i = \left(\hat{N}^{(1)}_i + \hat{N}^{(2)}_i - \hat{N}^{combined}_{(m+1)}\right)_+$
where $\hat{N}^{(s)}_i$ is the estimated count from sketch $s$
and $\hat{N}^{combined}_{(m+1)}$ is the $(m+1)^{th}$ smallest nonzero value obtained by  summing the estimated counts from the two sketches.
Previously, no merge operation existed for Deterministic Space Saving except to first convert it to a Misra-Gries sketch. 
Theorem \ref{thm:unbiased reduction} shows that by replacing the pairwise randomization with priority sampling or some other sampling procedure still allows one to obtain an Unbiased Space Saving merge that can preserve the expected count in the sketch rather than biasing it downward.

The trade-off required for such an unbiased merge operation is that the sketch may detect fewer of the top items by frequency than the biased Misra-Gries merge. Rather than truncating and preserving more of the "head" of the distribution, it must move mass from the tail closer to the head. This is illustrated in figure \ref{fig:merge}.

\begin{figure}
	\begin{center}	
		\begin{tabular}{cc}
			\includegraphics[width=2.2in, height=1.8in]{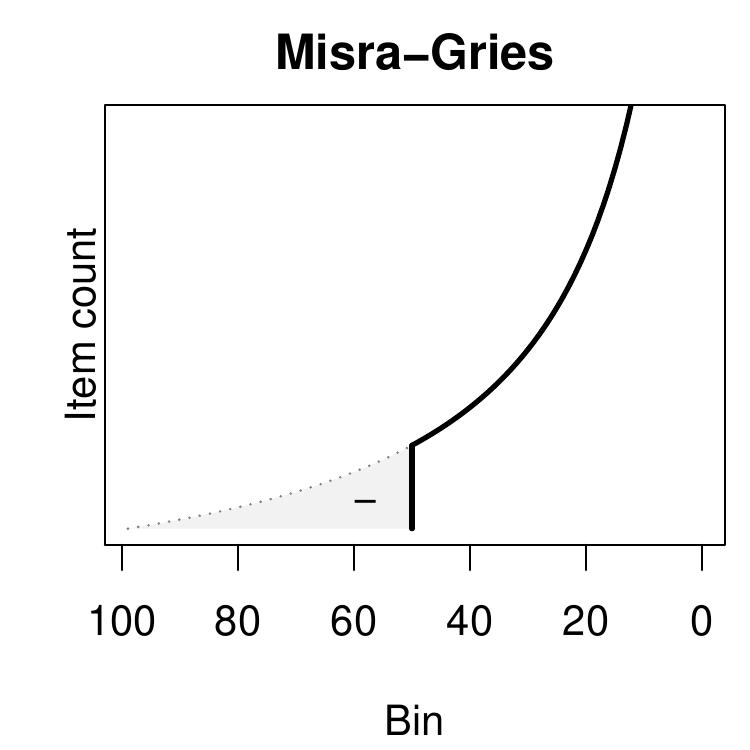} & 
			\includegraphics[width=2.2in, height=1.8in]{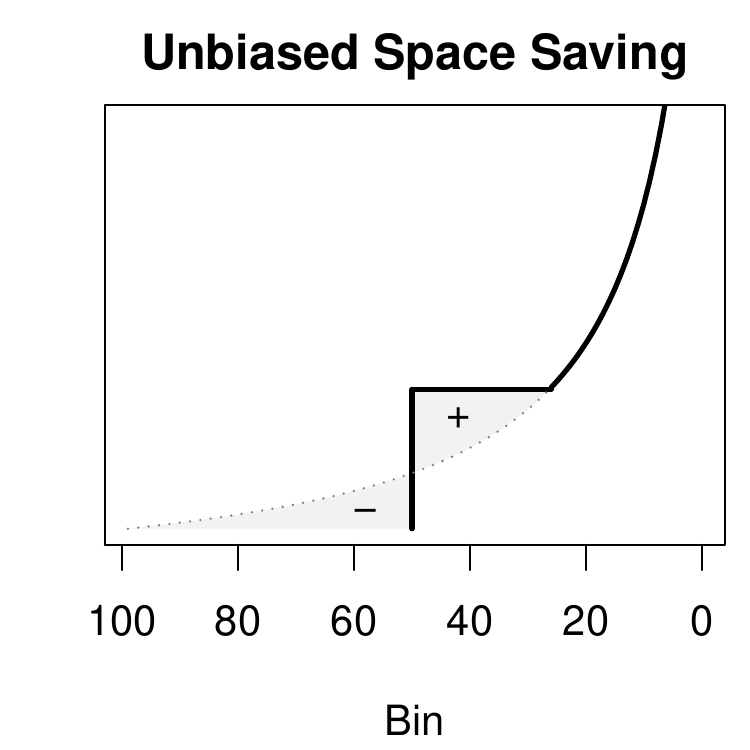} 
		\end{tabular}
	\end{center}
	\caption{In a merge operation, the Misra-Gries sketch simply removes mass from the extra bins with small count. Unbiased Space Saving moves the mass from infrequent items to moderately frequent items. It loses the ability to pick those items as frequent items in order to provide unbiased estimates for the counts in the tail.}
	\label{fig:merge}
\end{figure}

\section{Sketch Properties}
We study the properties of the space saving sketches here. These include provable asymptotic properties, variance estimates, heuristically and empirically derived properties, behavior on pathological and adversarial sequences, and costs in time and space. 
In particular, we prove that when the data is i.i.d., the sketch eventually includes all frequent items with probability 1 and that
the estimated proportions for these frequent items is consistent.
We prove there is a sharp transition between frequent items which are sampled with probability 1 eventually and infrequent items which are sampled with probability proportional to their sizes. This is also borne out in the experimental results where the observed inclusion probabilities match the theoretical ones and in estimation error where Unbiased Space Saving matches or even exceeds the accuracy of priority sampling.
In pathological cases, we demonstrate that Deterministic Space Saving completely fails at the subset estimation problem. Furthermore, these pathological sequences arise naturally. Any sequence where items' arrival rates change significantly over time forms a pathological sequence.
We show that we can derive a variance estimator as well. Since it works under pathological scenarios, the estimator is upward biased. However, we heuristically show that it is close to the variance for a PPS sample. This is confirmed in experiments as well.
For both i.i.d. and pathological cases, we examine the resulting empirical inclusion probabilities. Likewise, they behave similarly to a probability proportional to size or priority sample.

\subsection{Asymptotic consistency}
\label{sec:consistency}
Our main theoretical result for frequent item estimation states that the sketch contains all frequent items eventually on i.i.d. streams.
Thus it does no worse than Deterministic Space Saving asymptotically.
We also derive a finite sample bound in section \ref{sec:pathological}.
Furthermore, the guarantee states that the estimated proportion of times the item appears is strongly consistent and goes to 0. This is better than deterministic guarantees which only ensure that the error is within some constant.

Assume that items are drawn from a possibly infinite, discrete distribution with probabilities $p_1 \geq p_2 \geq \ldots$ and, without loss of generality, assume they are labeled by their index into this sequence of probabilities. Let $m$ be the number of bins and $t$ be the number of items processed by the sketch. We will also refer to $t$ as time.  Let
$\ical(t)$ be the set of items that are in the sketch at time $t$ and $Z_i(t) = 1(i \in \ical(t))$. To simplify the analysis, we will give a small further randomization by randomly choosing the smallest bin to increment when multiple bins share the same smallest count. 
Define an absolutely frequent item to be an item drawn with probability $> 1/m$ where $m$ is the number of bins in the sketch. 
By removing absolutely frequent items and decrementing the sketch size by 1 each time, the set of frequent item can be defined by the condition in corollary \ref{cor:joint frequent item} which depends only on the tail probability.
We first state the theorem and a corollary that immediately follows by induction.

\begin{theorem}
	\label{thm:frequent item}
	If $p_1 > 1/m$, then as the number of items $t \to \infty$,
	$Z_1(t) = 1$ eventually.
\end{theorem}

\begin{corollary}
	\label{cor:joint frequent item}
	If $p_i / \sum_{j \geq i} p_j > 1/(m-i+1)$ for all $i < \kappa$ and for some $\kappa < m$,
	then $Z_i(t) = 1$ for all $i < \kappa$ eventually. 
\end{corollary}

\begin{corollary}
	Given the conditions of corollary \ref{cor:joint frequent item},
	the estimate $\hat{p}_i(t) = \hat{N}_i(t) / t$ is strongly consistent for all $i < \kappa$ as $t \to \infty$.
\end{corollary}
\begin{proof}
	Suppose item $i$ becomes sticky after $t_0$ items are processed. After $t_0$, the number of times $i$ appears is counted exactly correctly. As $t \to \infty$, the number of times $i$ appears after $t_0$ will dominate the number of times it appears before $t_0$. By the strong law of large numbers, the estimate is strongly consistent.
\end{proof}

\begin{lemma}
	\label{lem:lower bound}
	Let $\alpha = \sum_{j > m} p_j$. For any $\alpha' < \alpha$,
	$N_{min}(t) > \alpha' t / m$ eventually as $t \to \infty$.
\end{lemma}
\begin{proof}
	Note that any item not in the sketch is added to the smallest bin. The probability of encountering an item not in the sketch is lower bounded by $\alpha$. Furthermore, by the strong law of large numbers, the actual number of items encountered that are not in the sketch must be $ > \alpha' t + m$ eventually.
	If there are $\alpha' t + m$ items added to the smallest bin, then with $m$ bins, $\hat{N}_{min}(t) > \alpha' t/m$. 
\end{proof}

We now give the proof of theorem \ref{thm:frequent item}. The outline of the proof is as follows. We first show that item $1$ will always reappear in the sketch if it is replaced. When it reappears, its bin will accumulate increments faster than the average bin, and as long as it is not replaced during this processes, it will escape and never return to being the smallest bin. Since the number of items that can be added before the label on the minimum bin is changed is linear in the size of the minimum bin, there is enough time for item $1$ to "escape" from the minimum bin with some constant probability. Even if it fails to escape on a given try, it will have infinitely many tries, so eventually it will escape.
\begin{proof}
	Trivially, $\hat{N}_{min}(t) \leq t / m$ since there are $m$ bins, and the minimum is less than the average number of items in each bin.
	If item $1$ is not in the sketch, then the smallest bin will take on $1$ as its label with probability $p_1 / (1+\hat{N}_{min}(t)) \geq m p_1 / (m+t)$. 
	Since  conditional on item 1 not being in the sketch, these are independent events, the second Borel-Cantelli lemma gives that  
	item $1$ is in the sketch infinitely often. 
	Whenever item $1$ is in the sketch, $\hat{N}_1(t) - t / m$ is a submartingale with bounded increments. 
	Furthermore, it can be lower bounded by an asymmetric random walk $\tilde{N}_1(t) - t/m$ where the expected increment is $\geq p_1 - 1/m$. Let $\epsilon = p_1 - 1/m$. Let $t_0$ be the time item 1 flips the label of the smallest bin. 
	Lemma \ref{lem:lower bound} gives that the difference $t_0/m - \hat{N}_1(t_0) < t_0(1-\alpha') / m$ for any $\alpha' < \sum_{j > m} p_j$
	If item 1 is not displaced, then after $d = 2 t_0(1-\alpha') / m\epsilon$ additional rows,
	Azuma's inequality gives after rearrangement, $P(\hat{N}_i(t_0 + d) - (t_0 + d) / m < 0) \leq P(\hat{N}_i(t_0 + d) - \hat{N}_i(t_0) - d/m - d \epsilon < -d \epsilon / 2)  < exp(-d \epsilon^2 /8) < exp(-\epsilon (1-\alpha')/4m)$. 
	The probability that item 1 is instead displaced during this time is $< d / (d+\alpha' t_0)$ which can be simplified to some positive constant that does not depend on $t_0$.
	In other words, there is some constant probability $\gamma$ such that item 1 will go from being in the smallest bin to a value greater than the mean. 
	From there, there is a constant probability that the bounded sub-martingle $\hat{N}_i(t_0 + d + \Delta) - (t_0+d+\Delta)/m$ never crosses back to zero or below. 
	Since item 1 appears infinitely often, it must either become sticky or there are infinitely many 0 upcrossing for $\hat{N}_1(t) - t/m$. In the latter case, there is a constant probability $\rho > 0$ that lower bounds the probability the item becomes sticky. Thus a geometric random variable lower bounds the number of tries before item $i$ "sticks," and it must eventually be sticky.
\end{proof}

\subsection{Approximate PPS Sample}
\label{sec:pps tail}
We prove that for i.i.d. streams, Unbiased Space Saving approximates a PPS sample and does so without the expensive pre-aggregation step. This is born out by simulations as, surprisingly, it often empirically outperforms priority sampling on computationally expensive, pre-aggregated data. Since frequent items are included with probability 1, we consider only the remaining bins and the items in the tail.

\begin{lemma}
	\label{lem:bin vs mean}
	Let $B_i$ denote the count in the $i^{th}$ bin.
	If $p_1 < 1/m$ then $B_i(t) - t/m < (\log t)^2 + 1$ eventually.
\end{lemma}
\begin{proof}
	If $B_i(t) > t/m$ then $B_i(t)$ is not the smallest bin. 
	In this case, the expected difference after 1 time step is bounded above by $\delta \defn p_1 - 1/m < 0$. Consider a random walk $W(t)$ with an increment of $1 - 1/m$ with probability $p_1$ and $-1/m$ otherwise. By Azuma's inequality, 
	if it is started at time $t-s$ at value $1$ then the probability it exceeds $(\log t)^2$ is bounded by 
	$P(W(t) - s / m - 1 > c(t) + \delta s) < exp(-(c(t) + \delta s)^2 / 2s)$. Since for $B_i(t) - t/m$ to be $> c(t)$, it must upcross 0 at some time $t-s$, maximizing over $s$ gives an upper bound on the probability $B_i(t) -t/m > c(t)$.
	It is easy to derive that $s = c(t) / \delta$ is the maximizer
	and the probability is bounded by $exp(-\delta c(t))$. When
	$c(t) = (\log t)^2$, $\sum_{t=1}^\infty exp(-\delta c(t)) < \infty$, and the conclusion holds by the Borel-Cantelli lemma.
\end{proof}

\begin{lemma}
	If $p_1 < 1/m$	then $0 \leq t/m - \hat{N}_{min} \leq m (\log t)^2 + m$ and $0 \leq \hat{N}_{max} - t/m \leq (\log t)^2 + 1$ eventually
\end{lemma}
\begin{proof}
	Since there are finitely many bins, by the lemma \ref{lem:bin vs mean}, $0 \leq \hat{N}_{max} - t/m \leq (\log t)^2 + 1$ eventually.	The other inequality holds since $t/m - \hat{N}_{min} < m(\hat{N}_{max} - t/m)$
\end{proof}

\begin{theorem}
	\label{thm:pps}
	If $p_1 < 1/m$, then the items in the sketch converge in distribution to a PPS sample.
\end{theorem}
\begin{proof}
	The label in each bin is obtained by reservoir sampling. Thus it is a uniform sample on the rows that go into that bin.
	Since all bins have almost exactly the same size $t/m + O( (\log t)^2)$, it follows that item $i$ is a label with probability $p_i + O((\log t)^2/t)$. 	
\end{proof}

The asymptotic analysis of Unbiased Space Saving splits items into two regimes. Frequent items are in the sketch with probability 1 and the associated counts are nearly exact. The threshold at which frequent and infrequent items are divided is given by corollary \ref{cor:joint frequent item} and is the same as the threshold in the merge operation shown in figure \ref{fig:merge}.
The counts for infrequent items in the tail are all  $\hat{N}_{min}(t)(1 + o(1))$.
The actual probability for the item in the bin is irrelevant since items not in the sketch will force the bin's rate to catch up to the rate for other bins in the tail.
Since an item changes the label of a bin with probability $1/B$ where $B$ is the size of the bin, the bin label is a reservoir sample of size 1 for the items added to the bin. Thus, the labels for bins in the tail are approximately proportional to their frequency.
Figure \ref{fig:inclusion probabilities} illustrates that the empirical inclusion probabilities match the theoretical ones for a PPS sample. The item counts are chosen to approximate a rounded $Weibull(5 \times 10^5, 0.15)$ distribution. This is a skewed distribution where the standard deviation is roughly $30$ times the mean.

\begin{figure}
	\begin{center}
		\includegraphics[width=4.5in]{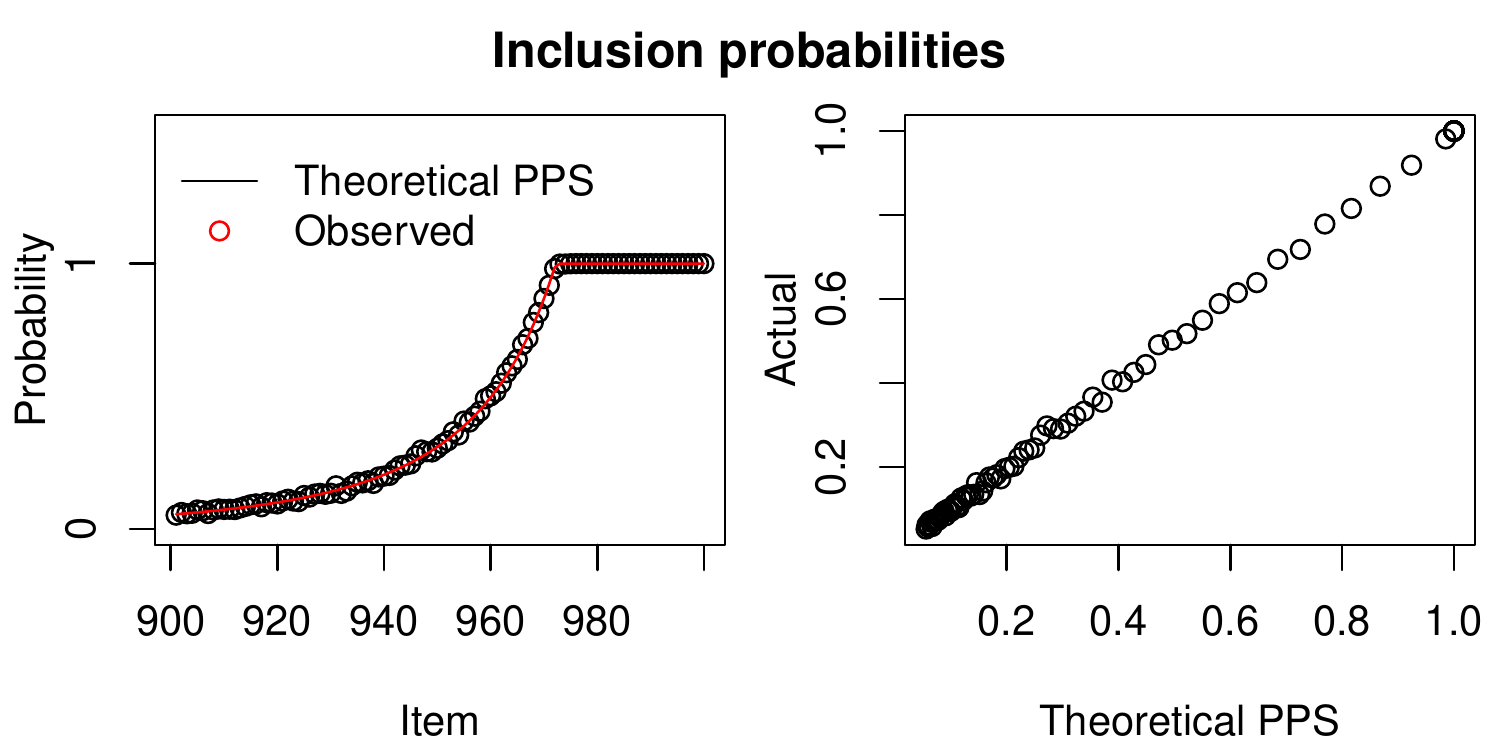}
	\end{center}
	\caption{The inclusion probability for each item empirically behaves like the inclusion probability for a probability proportional to size sample. This is also nearly equivalent to the inclusion probabilities for a priority sample. %
	}
	\label{fig:inclusion probabilities}
\end{figure}

We note, however, that the resulting PPS sample has limitations not present in PPS samples on pre-aggregated data. For pre-aggregated data, one has both the original value $x_i$ and the Horvitz-Thompson adjusted value $x_i / \pi_i$ where $\pi_i$ is the inclusion  probability. This allows the sample to compute non-linear statistics such as the population variance  which uses the second moment estimator $\sum_i x_i^2 Z_i / \pi_i$. With the PPS samples from disaggregated subset sum sketching, only the adjusted values are observed. 

\subsection{Pathological sequences}
\label{sec:pathological}
Deterministic Space Saving has remarkably low error when estimating the counts of frequent items \cite{cormode2008finding}. However, we will show that it fails badly when estimating subset sums when the data stream is not i.i.d.. Unbiased Space Saving performs well on both i.i.d. and on pathological sequences.

Pathological cases arise when an item's arrival rate changes over time rather than staying constant.  Consider a sketch with  2 bins. For a sequence of $c$ 1's, $c$ 2's, a single 3 and 4, the Deterministic Space Saving algorithm will always return 3 and 4, each with count $c+1$. By contrast,
Unbiased Space Saving will return 1 and 2 with probability $(1-1/c)^2 \approx 1$ when $c$ is large. Note that in this case, the count for each frequent item is slightly below the threshold that guarantees inclusion in the sketch, $c < n / 2$. This example illustrates the behavior for the deterministic algorithm. When an item is not in the "frequent item head" of the distribution then the bins that represent the tail pick the labels of the most recent items without regard to the frequency of older items.

We note that natural pathological sequences can easily occur. For instance, partially sorted data can naturally lead to such pathological sequences. This can occur from sketching the output of some join.
Data partitioned by some key where the partitions are processed in order is another case. We explore this case empirically in section \ref{sec:experimental}. 
Periodic bursts of an item followed by periods in which its frequency drops below the threshold of guaranteed inclusion are another example. 
The most obvious "pathological" sequence is the case where every row is unique. The Deterministic Space Saving sketch always consists of the last $m$ items rather than a random sample, and no meaningful subset sum estimate can be derived.

For Unbiased Space Saving, we show that even for non-i.i.d. sequences, it essentially never has an inclusion probability worse than simple random sampling which has inclusion probability $1 - (n_{tot} - n_i)_{m} / (n_{tot})_{m} \approx 1 - (1-n_i  / n_{tot})^m$ where $(x)_m$ denotes the $m^{th}$ falling factorial.

\begin{theorem}
	\label{thm:worst case}
	An item $i$ occurring $n_i$ times has worst case inclusion probability $\pi_i \geq 1 - (1-n_i  / n_{tot})^m$. An item with asymptotic frequency 
	$n_i = \alpha n / m + o(n/m)$ has an inclusion probability 
	$\pi_i \geq 1 - e^{-\alpha} + o(1)$ as $n,m \to \infty$.
\end{theorem}
\begin{proof}
	Whether an item is in the sketch depends only on the sequence of additions to the minimum sized bin. Let $T_{b}$ be last time an item is added to bin $b$ while it is the minimum bin. Let $C_{i,b}$ be the number of times item $i$ is added to bin $b$ by time $T_{b}$
	and $L_b$ be the count of bin $b$ at time $T_b$.
	Item $i$ is not the label of bin $b$ with probability $1 - C_{i,b} / L_b$,
	and it is not in the sketch with probability $\prod_b (1-C_{i,b} / L_b)$.
	Note that for item $i$ to not be in the sketch, the last occurrence of $i$ must have been added to the minimum sized bin.
	Thus, maximizing this probability under the constraints that $\sum_b L_b \leq n$ and $\sum_b C_{i,b} = n_i$ gives an upper bound on $1-\pi_i$ and yields the stated result.
\end{proof}

We note that the bound is often quite loose. It assumes a pathological sequence where the minimum sized bin is as large as possible, namely $L_b = n_{tot} / m$. If $L_b \leq \gamma n_{tot} / m$, the asymptotic bound would be $\pi_i \geq 1 - e^{-\alpha / \gamma} + o(1)$.

At the same time, we note that the bound is tight in the sense that one can construct a pathological sequence that achieves the upper bound.
Consider the sequence consisting of $n_{tot} - n_i$ distinct items followed by item $i$ for $n_i$ times with $n_i$ and $n_{tot}$ both being multiples of $m$. It is easy to see that the only way that item $i$ is not in the sketch is for it no bin to ever take on label $i$ and for the bins to all be equal in size to the minimum sized bin.  The probability of this event is equal to the given upper bound.

Although Deterministic Space Saving is poor on pathological sequences, we note that if data arrives in uniformly random order or if the data stream consists of i.i.d. data, one expects the Deterministic Space Saving algorithm to share similar unbiasedness properties as the randomized version as in both cases the label for a bin can be treated roughly as a uniform random choice out of the items in that bin.

\subsection{Variance}
In addition to the count estimate, one may also want an estimate of the variance. 
In the case of i.i.d. streams, this is simple since it forms an approximate PPS sample. Since the inclusion of items is negatively correlated,
a fixed size PPS sample of size $m$ has variance upper bounded by
\begin{align}
\label{eqn:PPS variance}
\var_{PPS}(\hat{N}_i) \leq \alpha_i n_i (1- \pi_i ).
\end{align}
When the marginal sampling probabilities $\pi_i = \min\{ 1, \alpha n_i\} $ are small, this upper bound is nearly exact.
For the non-i.i.d. case, we provide a coarse upper bound.
Since $\hat{N}_i(t) - n_i(t)$ is a martingale as shown in theorem \ref{thm:unbiased reduction},
the quadratic variation process taking the squares of the increments
$\sum_t (\hat{N}_i(t+1) - \hat{N}_i(t) - n_i(t+1) + n_i(t))^2$ yields an unbiased estimate of the variance.
There are only two cases where the martingale increment is non-zero: the new item is $i$ and $i$ is not in the sketch or
the new item is not $i$ and item $i$ is in the smallest bin.
In each case the expected squared increment is $\hat{N}_{min}(t) -1$ since the updated value is $1 + \hat{N}_{min}(t) \tilde{Z}_t$ where $\tilde{Z}_t \sim Bernoulli(N_{min}(t)^{-1})$.  Let $\tau_i$ be the time when item $i$ becomes "sticky." That is the time at which a bin acquires label $i$ and never changes afterwards. If item $i$ does not become sticky, then $\tau_i = n$.
Define $\kappa_i = n_i(\tau_i)$. It is the number of times item $i$ is added up to when it becomes sticky. This leads to the following upper bound on the variance
\begin{align}
\var(\hat{N}_i) &\leq 
\sum_{j = 0}^{\kappa_i} \E \left[\left(\hat{N}_{min} - \left\lfloor \frac{j}{m} \right\rfloor\right)_+ - 1\right]\\
&\leq \E (\hat{N}_{min} \kappa_i). \label{eqn:variance bound}
\end{align}
We note that the same variance argument holds when computing a further aggregation to estimate $n_S = \sum_{i \in S} n_i$ for a set of items $S$. In this case $\kappa_S$ is the total number of times items in $S$ are added to the sketch excluding the deterministic additions to the final set of "sticky" bins. 

To obtain a variance estimate for a count, we plug in an estimate for $\hat{\kappa}_i$
into equation \ref{eqn:variance bound}. We use the following estimate 
\begin{align}
\hat{\kappa}_S &= \hat{N}_{min} C_S \\
\widehat{\var}(\hat{N}_i) &= \hat{N}_{min}^2 C_S \label{eqn:var estimate}
\end{align}
where $C_S$ is the greater of $1$ and the number of times an item in $S$ appears in the sketch.

The estimate $\hat{\kappa}_S$ is an upward biased estimate for $\kappa_S$.
For items with count $\leq \hat{N}_{min}$, one has no information about their relative frequency compared to other infrequent items. Thus, we choose the worst case as our estimate $\hat{\kappa}_S = \hat{N}_{min}$.
For items with count $> \hat{N}_{min}$, we also take a worst case approach for estimating $\kappa$.
Consider a bin with size $\leq V - 1$. The probability that an additional item will cause a change in the label is $1/V$. Since $\hat{N}_{min}$ is the largest possible "non-sticky" bin, it follows $\kappa_i - 1< Y$ where $Y \sim Geometric(1/\hat{N}_{min})$. Taking the expectation
given $\hat{N}_{min}$ gives the upward biased estimate $\hat{\kappa}_i = \hat{N}_{min} + 1$. In this case, we drop the $1$ for simplicity and because it is an overestimate.

We compare this variance estimate with the variance of a Poisson PPS sample and show that they are similar for infrequent items but adds an additional term for each frequent item in the worst case for Unbiased Space-Saving.
Note that in the i.i.d. scenario for Unbiased Space-Saving, $\E C_i = \pi_i \to n_i / \alpha$ and  $\hat{N}_{min}$ converges to $\alpha$. Plugging these into equation \ref{eqn:var estimate} gives a variance estimate of $\alpha n_i$ which differs only by a factor of $1-\pi_i$ from the variance of a Poisson PPS sample given in equation \ref{eqn:PPS variance}. For infrequent items, $\pi_i$ is typically small. For frequent items, a Poisson PPS sample has inclusion probability $1$ and zero variance. In this case, the worst case behavior for Unbiased Space Saving contributes the same variance as an infrequent item.

The similar behavior to PPS samples is also borne out by experimental results. Figure \ref{fig:pathological sd} shows that the variance estimate is often very accurate and close to the variance of a true PPS sample.

\subsection{Confidence Intervals}
As the inclusion of a specific item is a binary outcome, confidence intervals for individual counts are meaningless. However,
the variance estimate allows one to compute Normal confidence intervals when computing sufficiently large subset sums. Thus, a system employing the sketch can provide estimates for the error along with the count estimate itself. These estimates are valid even when the input stream is a worst case non-i.i.d. stream. Experiments in section \ref{sec:experimental} shows that these Normal confidence intervals have close to or better than advertised coverage whenever the central limit theorem applies, even for pathological streams. 

\subsection{Robustness}
For the same reasons it has much better behavior under pathological sequences, Unbiased Space Saving is also more robust to adversarial sequences than Deterministic Space Saving.
Theorem \ref{thm:adversarial} shows that by inserting an additional $n_{tot}$ items, one can force all estimated counts to 0, including estimates for frequent items, as long as they are not too frequent. This complete loss of useful information is a strong contrast to the theoretical and empirical results for Unbiased Space Saving which suggest that polluting a dataset with $n_{tot}$ noise items will simply halve the sample size, since it will still return a sample that approximates a PPS sample.

\begin{theorem}
	\label{thm:adversarial}
	Let $\mathbf{n}$ be a vector of $v$ counts with $n_{tot} = \sum_{i=1}^v n_i$ and $n_i < 2 n_{tot} / m$ for all $i \leq v$. There is a sequence of $2 n_{tot}$ rows such that item $i$ appears exactly $n_i$ times, but the Deterministic Space Saving sketch returns an estimate of $0$ for all items $i \leq v$. 
\end{theorem}
\begin{proof}
	Sort the items from most frequent to least frequent. Add $n_{tot}$ additional distinct items.
	The resulting deterministic sketch will consist only of the additional distinct items and each bin will have count $2n_{tot}/m \pm 1$.
\end{proof}

\subsection{Running time and space complexity}
The update operation is identical to the Deterministic Space Saving update except that it changes the label of a bin less frequently. Thus, each update can be performed in $O(1)$ time \cite{metwally2005spacesaving} when the stream summary data structure is used. In this case the space usage is $O(m)$ where $m$ is the number of bins.

\section{Experiments}
\label{sec:experimental}
We perform experiments with both simulations and real ad prediction data.
For synthetic data, we consider three cases: randomly permuted sequences, realistic pathological sequences for Deterministic Space Saving, 
and "pathological" sequences for Unbiased Space Saving.
For each we draw the count for each item using a Weibull distribution that is discretized to integer values. That is $n_i \sim Round(Weibull(k, \alpha))$ for item $i$. 
The discretized Weibull distribution is a generalization of the geometric distribution that allows us to adjust the tail of the distribution to be more heavy tailed. We choose it over the Zipfian or other truly heavy tailed distributions as few real data distributions have infinite variance. Furthermore, we expect our methods to perform better under heavy tailed distributions with greater data skew as shown in figure \ref{fig:iid}. 
For more easily reproducible behavior we applied the inverse cdf method $n_i = F^{-1}(U_i)$ where the $U_i$ are on a regular grid of $1000$ values rather than independent $Uniform(0,1)$ random variables. Randomly permuting the order in which individual rows arrive yields an exchangeable sequence which we note is identical to an i.i.d. sequence in the limit by de Finetti's theorem.
In each case, 
we draw at least $10,000$ samples to estimate the error.

For real data, we use a Criteo ad click prediction dataset
\footnote{\url{http://labs.criteo.com/2014/02/kaggle-display-advertising-challenge-dataset/}}. 
This dataset provides a sample of 45 million ad impressions. Each sample includes the outcome of whether or not the ad was clicked as well as multiple integer valued and categorical features. We do not randomize the order in which data arrives in this case.
We pick a subset of 9 of these features. There are over 500 million possible tuples on these features and many more possible filtering conditions.

The Criteo dataset provides a natural application of the disaggregated subset sum problem. Historical clicks are a powerful feature in click prediction \cite{richardson2007predicting}, \cite{hillard2010improving}.
While the smallest unit of analysis is the $ad$ or the $(user, ad)$ pair,
the data is in a disaggregated form with one row per impression. Furthermore,
since there may not be enough data for a particular ad, the relevant click prediction feature may be the historical click through rate for the advertiser or some other higher level aggregation. 
Past work using sketches to estimate these historical counts \cite{shrivastava2016time} include the CountMin counting sketch as well as the Lossy Counting frequent item sketch. 

To simulate a variety of possible filtering conditions, we draw random subsets of 100 items to evaluate the randomly permuted case. As expected, subsets which mostly pick items in the tail of the distribution and have small counts also have estimates with higher relative root mean squared error. The relative root mean squared error (RRMSE) is defined as $\sqrt{MSE} / n_S$ where $n_S$ is the true subset sum. For unbiased estimators this is equivalent to 
$\sigma_S / n_S$ where $\sigma_S$ is the standard deviation of the estimator.
Note that an algorithm with $c$ times the root mean squared error of a baseline algorithm typically requires $c^2$  times the space as the variance, not the standard deviation, scales linearly with size.

We compare out method to uniform sampling of items using the bottom-k sketch, priority sampling, and Deterministic Space Saving. Although we do not directly compare with sample and hold methods, we note that figure 2 in \cite{cohen2007sketching} shows that sample and hold performs significantly worse than priority sampling.

\begin{figure}
	\includegraphics[width=4.4in]{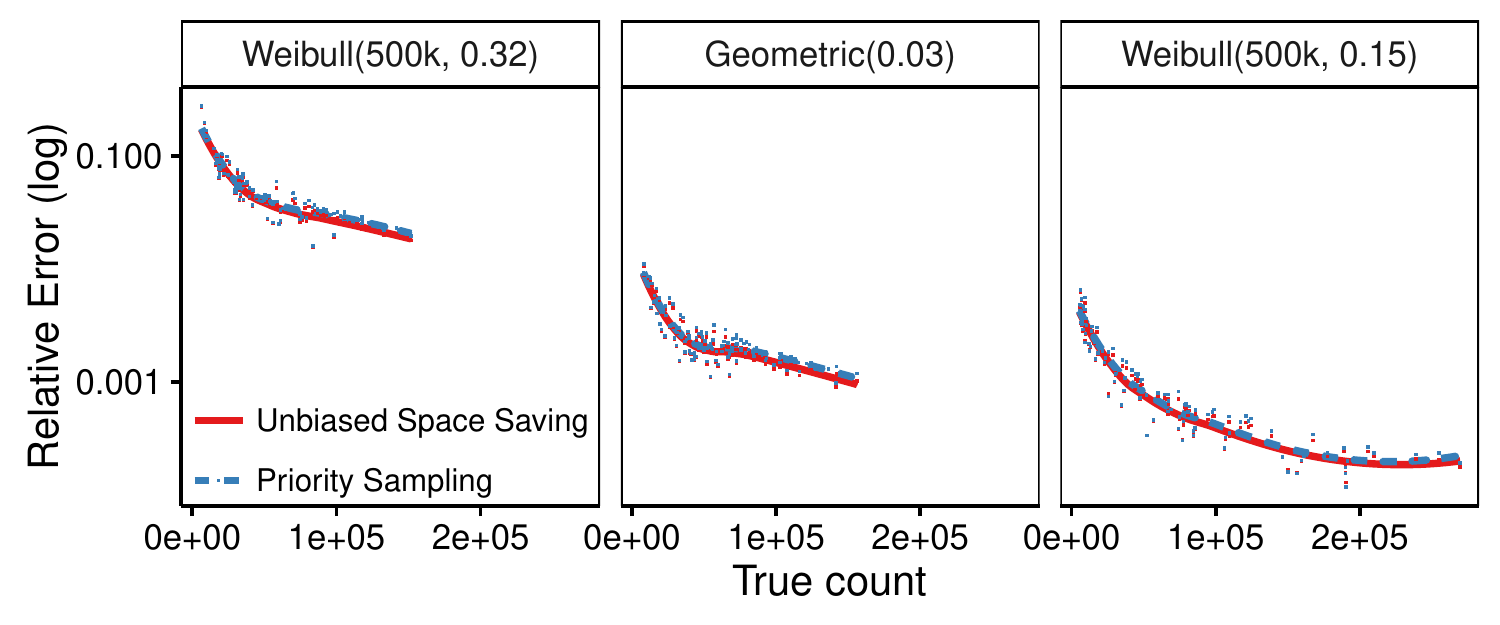} 
	\caption{The sketch accuracy improves when the skew is higher and when more and larger bins are contained in the subset. The number of bins is 200. }
	\label{fig:iid}
\end{figure}

\begin{figure}
	\includegraphics[width=4.4in]{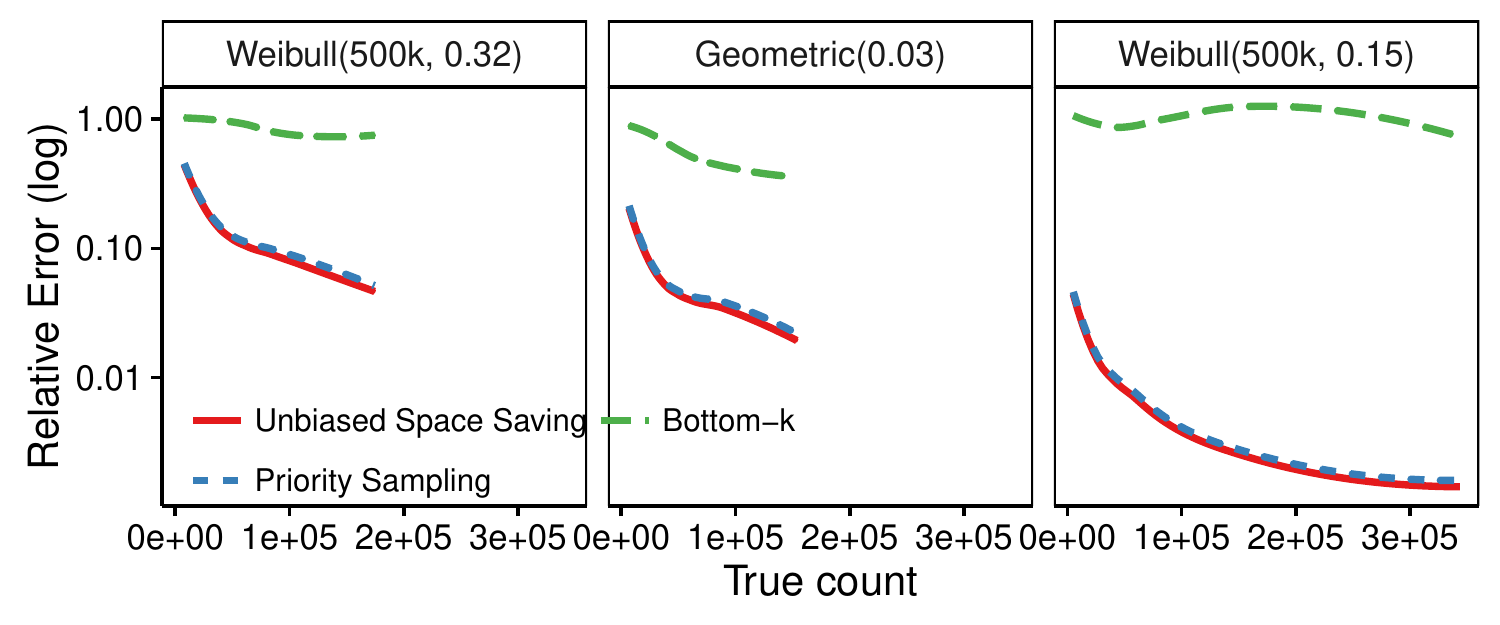} 
	\caption{Unbiased Space Saving performs orders of magnitude better than uniform sampling of items (Bottom-k) in the smoothed plot of relative error versus the true count. With 100 bins, the error is higher than with 200 bins given in figure \ref{fig:iid} but the curve is qualitatively similar.  }
	\label{fig:bad}
\end{figure}

Surprisingly, figure \ref{fig:Priority vs Space Saving}  shows that Unbiased Space Saving performs better than priority sampling even though priority sampling is applied on pre-aggregated data. We are unsure as to the exact reason for this. However, we note that, unlike Unbiased Space Saving, priority sampling does not ensure the total count is exactly correct. A priority sample of size $100$ when all items have the same count will have relative error of $\approx 10\%$ when estimating the total count. 

This added variability in the threshold and the relatively small sketch sizes for the simulations on i.i.d. streams may explain why Unbiased Space Saving performs even better than what could be considered close to a "gold standard" on pre-aggregated data.

\begin{figure}
	\begin{center}
	\includegraphics[width=4.4in]{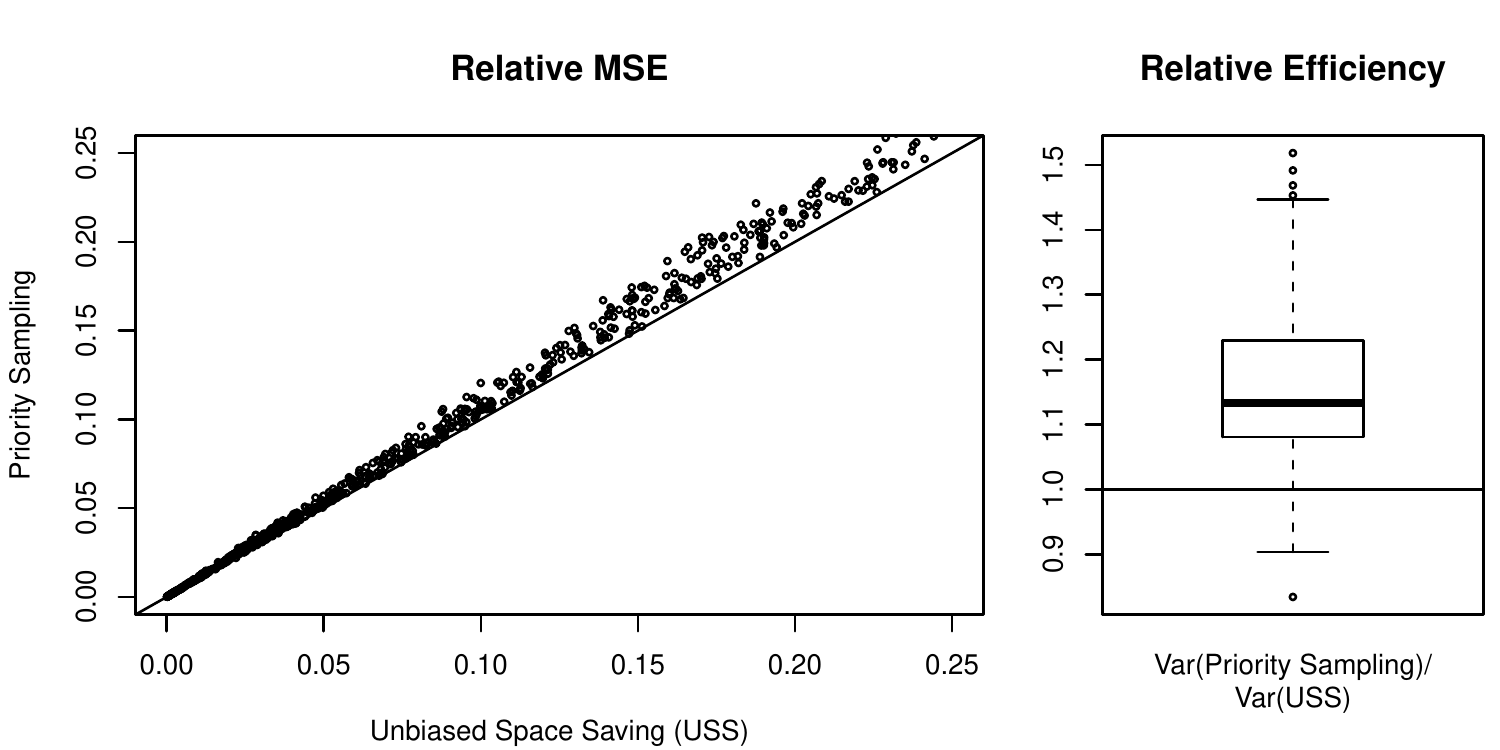}
	\end{center}
	\caption{Unbiased Space Saving performs slightly better than priority sampling on the synthetic data despite priority sampling using pre-aggregated data rather than the raw unaggregated data stream.}
	\label{fig:Priority vs Space Saving}
\end{figure}

\begin{figure}
	\includegraphics[width=4.4in]{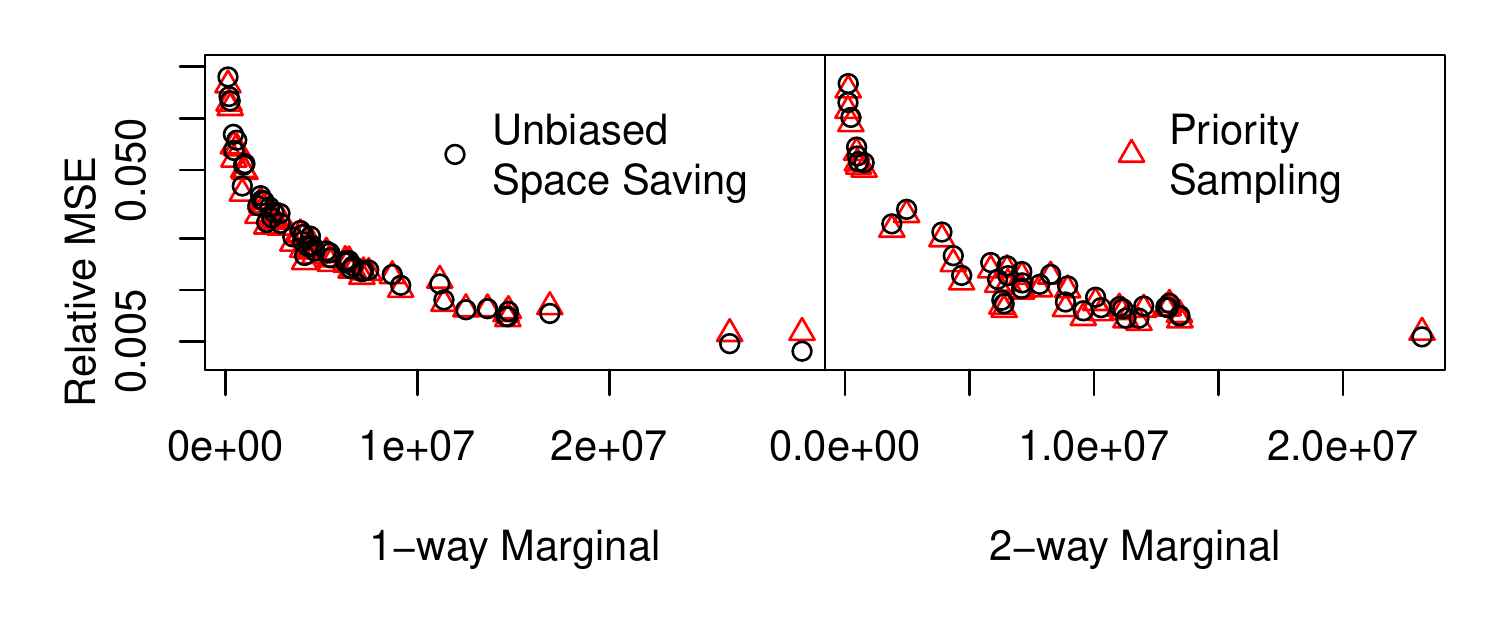} 
	\caption{The Unbiased Space Saving sketch is able to accurately compute 1 and 2 way marginals. The average relative mse for a marginal count that is between 100k and 200k is $< 5\%$ and for marginals containing more than half the data, the mean squared error drops to under $0.5\%$. It performs similarly to priority sampling.}
	\label{fig:criteo}
\end{figure}

\subsection{Pathological cases and variance}

For pathological sequences we find that Unbiased Space Saving performs well in all cases
while Deterministic Space Saving gives unacceptably large errors even for reasonable non-i.i.d. sequences.
First we consider a pathological sequence for Deterministic Space Saving. This sequence is generated by splitting the sequence into two halves. Each half is an independent i.i.d. stream from a discretized Weibull frequency distribution. This is a natural scenario as the data may be randomly partitioned into blocks, for example, by hashed user id, and each block is fed into the sketch for summarization.
As shown in figure \ref{fig:pathological}, Deterministic Space Saving completely ignores infrequent items in the first half of the stream, resulting in large bias and error. In this case, the sketches used are small with only $100$ bins, and the disparity would only increase with larger sketches and streams where the bias of Deterministic Space Saving remains the same but the variance decreases for Unbiased Space Saving. 

\begin{figure}
	\begin{center}	
		\begin{tabular}{cc}
			\includegraphics[width=2.1in]{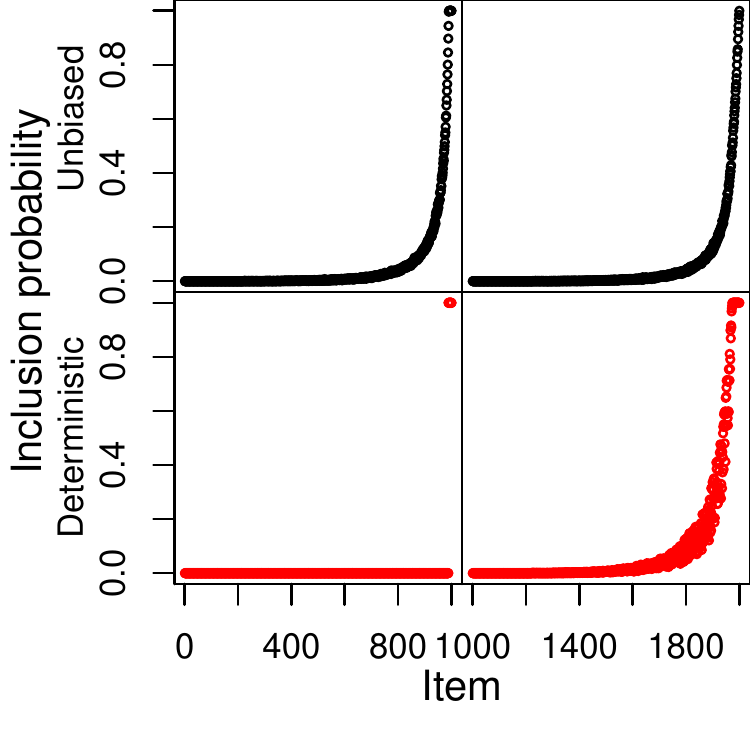} & 
			\includegraphics[width=2.1in]{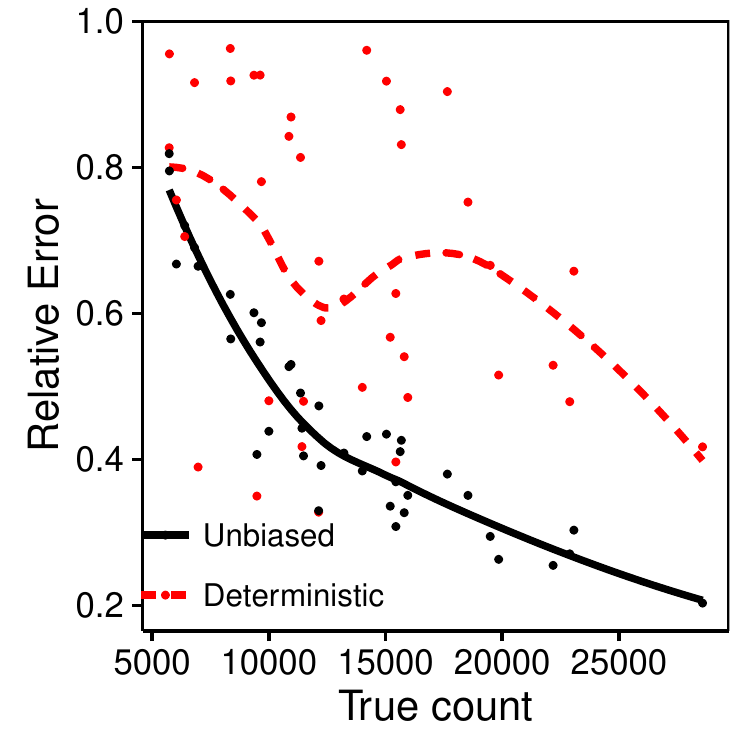} 
		\end{tabular}
	\end{center}
	\caption{Left: Items 1 to 1,000 only appear in the first half of the stream. The inclusion probabilities for a pathological sequence still behave like a PPS sample for Unbiased Space Saving, but only the frequent items in the first half are sampled under Deterministic Space Saving.  Right: As a result, Deterministic Space Saving is highly inaccurate when querying items in the first half of the stream.}
	\label{fig:pathological}
\end{figure}

The types of streams that induce worst case behavior for Deterministic and Unbiased Space Saving are different. 
For Unbiased Space Saving, we consider a sorted stream arranged in ascending order by frequency. Note that the reverse order where the largest items occur first gives an optimally favorable stream for Unbiased Space Saving. Every frequent item is deterministically included in the sketch, and the count is exact. The sequence consists of $10^5$ distinct items and $10^9$ rows where the item counts are from a discretized Weibull distribution. We use $10,000$ bins in these experiments.
To evaluate our method, we partition the sequence into 10 epochs containing an equal number of distinct items and estimate the counts of items from each epoch. 
We find in this case our variance estimate given in equation \ref{eqn:var estimate} yields an upward biased estimate of the variance as expected. Furthermore, it is accurate except for very small counts and the last items in a stream.
Figure \ref{fig:pathological CI} shows the true counts and the corresponding 95\% confidence intervals computed as $\hat{N}_S \pm 1.96 \widehat{\var}(\hat{N}_{S})$. In epochs 4 and 5, there are on average roughly 3 and 13 items in the sample, and the asymptotic properties from the central limit theorem needed for accurate normal confidence intervals have not or are not fully manifested. For epochs 1 and 2, the upward bias of the variance estimate gives 100\% coverage despite the central limit theorem not being applicable.
The coverage of a confidence interval is defined to the the probability the interval includes the true value. A 95\% confidence interval should have almost exactly 95\% coverage.
Lower coverage represents an underestimation of variability or risk. 
Less harmful is higher coverage, which represents an overly conservative estimation of variability.

We note that the behavior of Deterministic Space Saving is easy to derive in this case. The first 9 epochs have estimated count of 0 and the last epoch has estimated count $n_{tot} = 10^{9}$. Figure \ref{fig:RRMSE} shows that except for small counts, Unbiased Space Saving performs an order of magnitude better than Deterministic Space Saving. 

\begin{figure}
	\begin{center}	
		\begin{tabular}{cc}
			\includegraphics[width=2.1in]{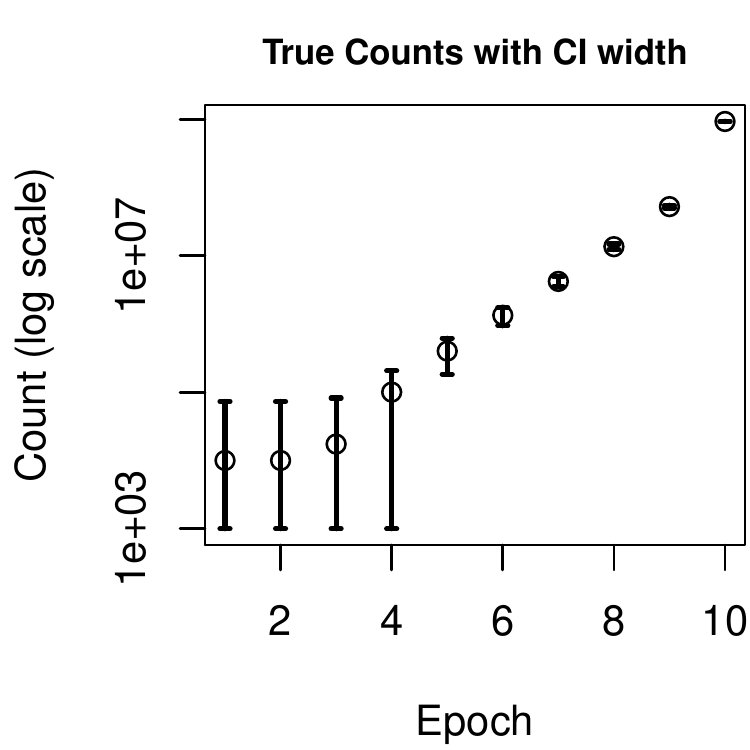} & 
			\includegraphics[width=2.1in]{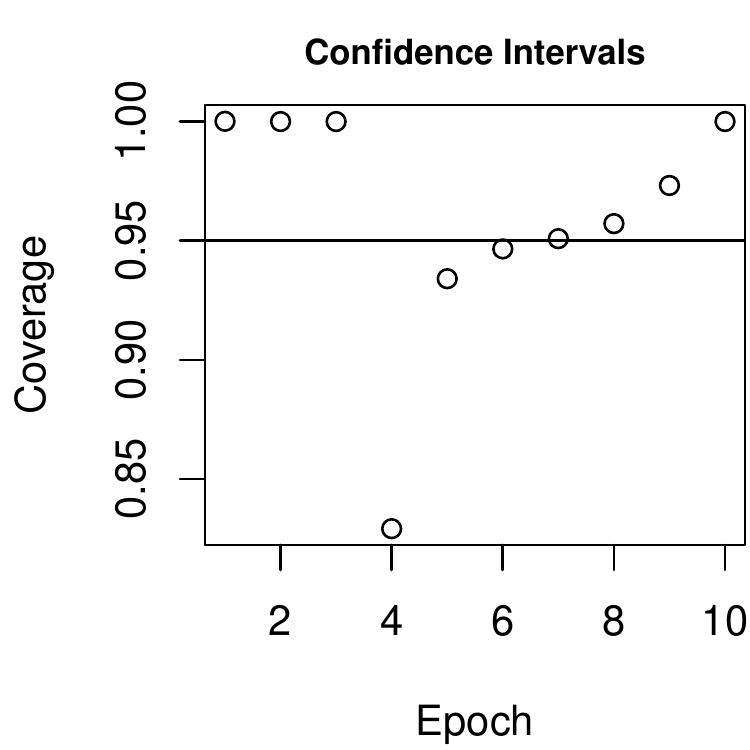} 
		\end{tabular}
	\end{center}
	\caption{Left: For a pathological sorted sequence, the true counts are given with bars indicating the average 95\% confidence interval width. For epochs 1 to 4, the intervals are truncated below as they extend past 0. Right: Normal confidence intervals generally deliver higher than advertised coverage. The exceptions lie in a regime where the variance estimate is accurate as shown in figure \ref{fig:pathological sd}, but the sample contains too few items from the epoch to apply the central limit theorem. }
	\label{fig:pathological sd}
\end{figure}

\begin{figure}
	\begin{center}	
		\begin{tabular}{cc}
			\includegraphics[width=2.1in]{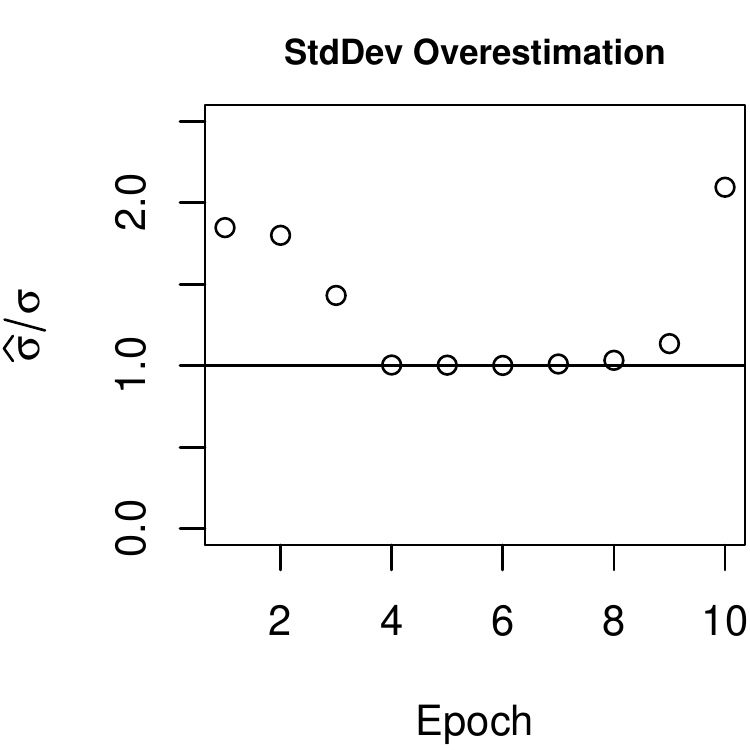} & 
			\includegraphics[width=2.1in]{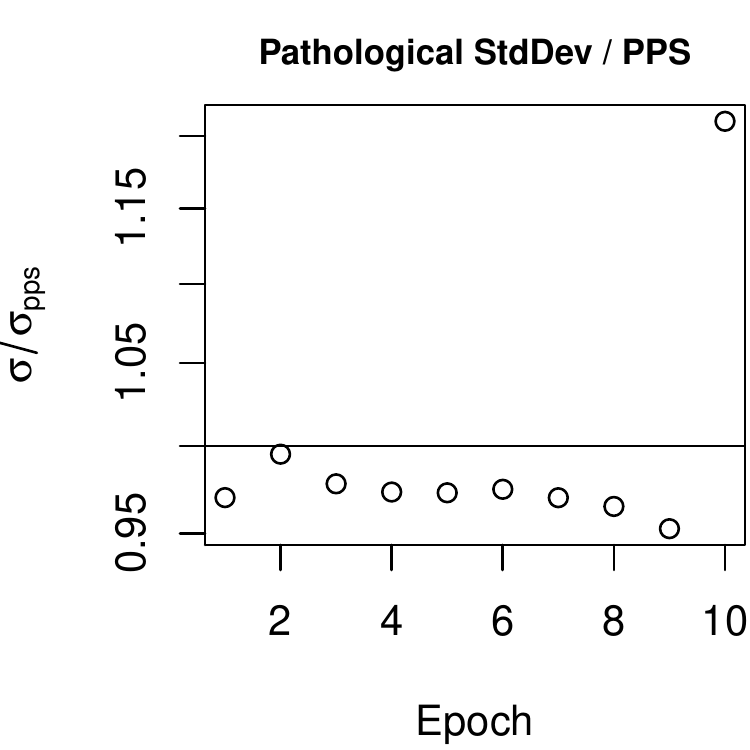} 
		\end{tabular}
	\end{center}
	\caption{Left: For pathological streams, the estimated standard deviation from equation \ref{eqn:var estimate} is shown to be accurate and match the true standard deviation for counts that are not too large or small. Right: Even for pathological streams, the variance closely matches the variance from a PPS sample. }
	\label{fig:pathological CI}
\end{figure}

\begin{figure}
		\begin{center}
	\includegraphics[width=2.5in]{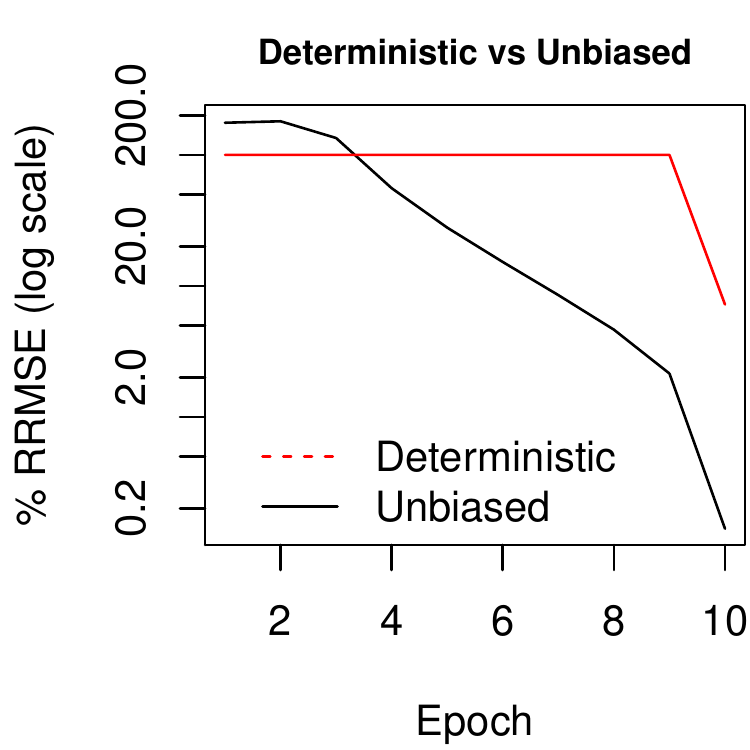} 
			\end{center}
	\caption{Deterministic Space Saving gives completely incorrect answers on all epochs. 
		For epochs 9 and 10, the error is 50x that of Unbiased Space Saving. For extremely small counts constituting $<  0.002\%$ of the total, the possibility of overestimation causes Unbiased Space Saving to have worse error compared to the 0 estimate given by Deterministic Space Saving.}
	\label{fig:RRMSE}
\end{figure}

\section{Conclusion}
We have introduced a novel sketch, Unbiased Space Saving, that answers both the disaggregated subset sum and frequent item problems and gives state of the art performance under all scenarios. Surprisingly, for the disaggregated subset sum problem, the sketch can outperform even methods that run on pre-aggregated data. 

We prove that asymptotically, it can answer the frequent item problem for i.i.d. sequences with probability 1 eventually. Furthermore, it gives stronger probabilistic consistency guarantees on the accuracy of the count than previous results for Deterministic Space Saving.  For non-i.i.d. streams, we show that Unbiased Space Saving still has attractive frequent item estimation properties and exponential concentration of inclusion probabilities to 1. 

For the disaggregated subset sum problem, 
we prove that the sketch provides unbiased results. 
For i.i.d. stream, we show that items selected for the sketch are sampled approximately according to an optimal PPS sample.
For non-i.i.d. streams we show that it empirically performs well and is close to a PPS sample even if given a pathological stream
for which Deterministic Space Saving fails badly on.
We derive a variance estimator for subset sum estimation and show that it is nearly equivalent to the estimator for a PPS sample. It is shown to be accurate on pathological sequences and yields confidence intervals with good coverage.

We study Unbiased Space Saving's behavior and connections to other data sketches. In particular, we identify the primary difference between many of the frequent item sketches is a slightly different operation to reduce the number of bins. We use that understanding to provide multiple generalizations to the sketch which allow it to be applied in distributed settings, handle weight decay over time, and adaptively change its size over time. This also allows us to compare Unbiased Space Saving to the family of sample and hold sketches that are also designed to answer the disaggregated subset sum problem. This allows us to also mathematically show that Unbiased Space Saving is superior.

\bibliographystyle{abbrvnat}
\bibliography{ling2} 

\end{document}